\documentclass{article}
\usepackage{aaai18}
\usepackage{times}
\usepackage{graphicx}
\usepackage{amsmath,amssymb,amsfonts}
\usepackage{amsthm}
\usepackage{stmaryrd}
\usepackage{wasysym}
\usepackage{xspace}
\usepackage{array}
\usepackage{booktabs}
\usepackage{enumitem}
\usepackage{footnote}
\usepackage{thm-restate}
\usepackage{comment}
\usepackage[colorinlistoftodos]{todonotes}

\usepackage{etoolbox}
\apptocmd{\sloppy}{\hbadness 10000\relax}{}{}

\newcommand{\ourtitle}{Metric Temporal Extensions of \textit{DL-Lite} and Interval-Rigid~Names}

\usepackage[pdfpagelabels,bookmarksnumbered,bookmarksopen,
bookmarks=false,
bookmarksopenlevel=1,
colorlinks=true,
citecolor=black, 
filecolor=black, 
linkcolor=black, 
menucolor=black,
urlcolor=black, 
pdftitle={\ourtitle}, 
pdfauthor={Veronika Thost},
pdfkeywords={Computational Complexity of Reasoning, Description Logics and Ontologies, Temporal Logics, Knowledge Representation Languages, Logics for Knowledge Representation}]
{hyperref}

\newtheorem{theorem}{Theorem}
\newtheorem{lemma}[theorem]{Lemma}
\newtheorem*{claim*}{Claim}

\newcommand{\lnext}{{\ocircle_F}}

\newcommand{\lalways}{\Box}
\newcommand{\luntil}{\,\mathcal{U}}

\newcommand{\aaxiom}{\ensuremath{\alpha}\xspace}

\newcommand{\as}{\ensuremath{\mathcal{W}}\xspace}

\newcommand{\bin}{\ensuremath{^{\mathsf{bin}}}\xspace}

\newcommand{\zinfty}{\ensuremath{^{0, \infty}}\xspace}

\newcommand{\mypar}[1]{\medskip\noindent\textbf{#1.}}
\newcommand{\csub}{\ensuremath{\mathbb{C}}\xspace} 
\newcommand{\fsub}{\ensuremath{\mathbb{F}}\xspace}

\newcommand{\Nbb}{\ensuremath{\mathbb{N}}\xspace}
\newcommand{\Zbb}{\ensuremath{\mathbb{Z}}\xspace}

\newcommand{\NRig}{\ensuremath{\mathsf{N_{Rig}}}\xspace}
\newcommand{\NIRig}{\ensuremath{\mathsf{N_{IRig}}}\xspace}
\newcommand{\ir}{\ensuremath{\mathsf{iRig}}\xspace}

\newcommand{\const}{\ensuremath{k_{\Phi}}\xspace}
\newcommand{\individuals}[1]{\mathsf{N_I}(#1)}
\newcommand{\cons}[1]{\mathbb{C}(#1)}

\newcommand{\EL}{\ensuremath{{\cal E\!L}}\xspace}
\newcommand{\ALC}{\ensuremath{{\cal ALC}}\xspace}

\newcommand{\NC}{\ensuremath{{\sf N_C}}\xspace}
\newcommand{\NI}{\ensuremath{{\sf N_I}}\xspace}
\newcommand{\NR}{\ensuremath{{\sf N_R}}\xspace}

\newcommand{\NP}{\textsc{NP}\xspace}
\newcommand{\PSpace}{\textsc{PSpace}\xspace}

\newcommand{\TwoExpTime}{\textsc{2-ExpTime}\xspace}
\newcommand{\NExpTime}{\textsc{NExpTime}\xspace}

\newcommand{\ExpSpace}{\textsc{ExpSpace}\xspace}
\newcommand{\TwoExpSpace}{\textsc{2-ExpSpace}\xspace}

\newcommand{\Fmc}{\ensuremath{\mathcal{F}}\xspace}

\newcommand{\Imc}{\ensuremath{\mathcal{I}}\xspace}
\newcommand{\Jmc}{\ensuremath{\mathcal{J}}\xspace}

\newcommand{\Tmc}{\ensuremath{\mathcal{T}}\xspace}

\newcommand{\Imf}{\ensuremath{\mathfrak{I}}\xspace}

\newcommand{\Rmf}{\ensuremath{\mathfrak{R}}\xspace}

\newcommand{\ind}{\ensuremath{_{\mathsf{ind}}}\xspace}
\newcommand{\indi}{\ensuremath{_{\mathsf{ind},i}}\xspace}

\newcommand{\horn}{\ensuremath{\textit{horn}}\xspace}
\newcommand{\krom}{\ensuremath{\textit{krom}}\xspace}
\newcommand{\bool}{\ensuremath{\textit{bool}}\xspace}

\newcommand{\DLLite}{\textit{DL-Lite}\xspace}

\newcommand{\DLLitecore}{\ensuremath{\textit{DL-Lite}_{\textit{core}}}\xspace}
\newcommand{\DLLitebool}{\ensuremath{\textit{DL-Lite}_{\textit{bool}}}\xspace}

\newcommand{\DLLitekrom}{\ensuremath{\textit{DL-Lite}_{\textit{krom}}}\xspace}

\newcommand{\DLLitehorn}{\ensuremath{\textit{DL-Lite}_{\textit{horn}}}\xspace}
\newcommand{\DLLitec}{\ensuremath{\textit{DL-Lite}_{c}}\xspace}

\newcommand{\LTLALC}{\text{\upshape{LTL}}\ensuremath{_\ALC}\xspace}

\newcommand{\QTL}[1]{\text{\upshape{LTL}}\ensuremath{\bin_{#1}}\xspace}
\newcommand{\QTLzi}[1]{\text{\upshape{LTL}}\ensuremath{\zinfty_{#1}}\xspace}

\newcommand{\LTLbin}{\text{\upshape{LTL}}\ensuremath{\bin}\xspace}

\newcommand{\LTLbinALC}{\QTL\ALC}

\newcommand{\LTLziALC}{\QTLzi\ALC}

\newcommand{\LTLDLL}{\text{\upshape{LTL}}\ensuremath{_{\DLLite}}\xspace}
\newcommand{\LTLDLLc}{\text{\upshape{LTL}}\ensuremath{_{\DLLite_c}}\xspace}
\newcommand{\LTLDLLbool}{\text{\upshape{LTL}}\ensuremath{_{\DLLitebool}}\xspace}
\newcommand{\LTLDLLkrom}{\text{\upshape{LTL}}\ensuremath{_{\DLLitekrom}}\xspace}

\newcommand{\LTLDLLhorn}{\text{\upshape{LTL}}\ensuremath{_{\DLLitehorn}}\xspace}

\newcommand{\LTLbinDLL}{\QTL\DLLite}
\newcommand{\LTLbinDLLc}{\QTL\DLLitec}
\newcommand{\LTLbinDLLbool}{\QTL\DLLitebool}
\newcommand{\LTLbinDLLkromhorn}{\QTL{\DLLite_{\textit{krom/horn}}}}
\newcommand{\LTLbinDLLkrom}{\QTL\DLLitekrom}
\newcommand{\LTLbinDLLhorn}{\QTL\DLLitehorn}

\newcommand{\LTLziDLLc}{\QTLzi\DLLitec}
\newcommand{\LTLziDLLbool}{\QTLzi\DLLitebool}

\newcommand{\aform}{\ensuremath{\Phi}\xspace}
\newcommand{\aformtwo}{\ensuremath{\Psi}\xspace}
\newcommand{\x}[1]{\ensuremath{\mathtt{#1}}\xspace}
\newcommand{\tabthm}{Thm.}
\newcommand{\citethm}{Thm.~}
\newcommand{\citelem}{Lem.~}
\newcommand{\pa}[1]{\ensuremath{#1^{\smash{\mathrm{p}}}}\xspace}
\newcommand{\elone}{\ensuremath{d}\xspace}
\newcommand{\eltwo}{\ensuremath{e}\xspace}

\newcommand{\aint}{\ensuremath{\mathcal{I}}\xspace}
\newcommand{\atint}{\ensuremath{\mathfrak{I}}\xspace}

\newcommand{\lnextf}{\ensuremath{\ocircle_F}\xspace}

\newcommand{\lprevp}{\ensuremath{\ocircle_P}\xspace}
\DeclareMathOperator*{\lsince}{\mathcal{S}}
\DeclareMathOperator*{\isince}{\ensuremath{\mathcal{S}^\textit{I}}\xspace}
\DeclareMathOperator*{\iuntil}{\ensuremath{\mathcal{U}^\textit{I}}\xspace}
\newcommand{\isincea}[1]{\ensuremath{\mathcal{S}^{#1}}\xspace}
\newcommand{\iuntila}[1]{\ensuremath{\mathcal{U}^{#1}}\xspace}
\newcommand{\ialways}[1]{\ensuremath{\Box^{#1}}\xspace}
\newcommand{\ialwaysf}[1]{\ensuremath{\Box^{#1}_F}\xspace}
\newcommand{\lalwaysf}{\ensuremath{\Box_F}\xspace}
\newcommand{\lalwaysp}{\ensuremath{\Box_P}\xspace}
\newcommand{\ieventf}[1]{\ensuremath{\smash{\Diamond_F^{#1}}}\xspace}

\newcommand{\leventp}{\ensuremath{\Diamond_P}\xspace}
\newcommand{\levent}{\ensuremath{\Diamond}\xspace}

\newcommand{\NRM}{{\ensuremath{\mathbb{R}}}\xspace}
 \includecomment{tr}
 \excludecomment{paper}
	\title{\ourtitle}
\author{Veronika Thost\\
MIT-IBM Watson AI Lab, IBM Research\\
veronika.thost@ibm.com}%

\begin{tr}
\end{tr}
\begin{paper}
\end{paper}

\begin{document}
\maketitle
\begin{abstract}
%
 The \DLLite description logics allow for modeling domain knowledge on top of databases and 
 for efficient reasoning. 
 We focus on metric temporal extensions of \DLLitebool and its fragments, and study the complexity of satisfiability.
 In particular, we 
 investigate the influence of rigid and interval-rigid symbols, which allow for modeling knowledge that remains valid over (some) time. We show that especially the latter add considerable expressive power in many logics, but they do not always increase complexity.
\end{abstract}

%

\section{Introduction}
The \DLLite 
description logics 
allow for representing conceptual data models like UML class and ER diagrams 
formally  \cite{dllfamily}. 
%
%
%
In \DLLite, \emph{concept inclusions}~(CIs) can, for instance, capture the fact that every master's student is a student and that every student is enrolled: 
\begin{align*}
\x{MasterStudent}\sqsubseteq{}& \x{Student}\\
\x{Student}\sqsubseteq{}&\exists\x{EnrolledIn}
\intertext{%
	Here, $\x{MasterStudent}$ and $\x{Student}$ are \emph{concept names} that represent the sets of all (master's) students; $\x{EnrolledIn}$ is a \emph{role name} representing a binary
	relation connecting students to degree programs; and $\exists\x{EnrolledIn}$
	refers to the domain of that relation.
	%
	Since knowledge is often temporal, several temporal extensions of \DLLite have been investigated \cite{AKLWZ-TIME07:temporalising,AKRZ-TOCL14:cookbook,AKKRWZ-IJCAI15:omtqs,BoLT-JWS15,BoT-GCAI15,THO-DL15,dissthost}, which allow some qualitative operators of \emph{linear temporal logic} (LTL) 
	to occur within the axioms and/or for combining the axioms with such operators. For example, we may actually want to say that every student enrolled \emph{at some time in the past} ($\leventp$) and, after enrollment, pays the fee in the \emph{next} ($\lnextf$) month:} 
\x{Student}\sqsubseteq{}&\ \leventp\exists\x{EnrolledIn}\\ \exists \x{EnrolledIn}\sqsubseteq{}&\lnextf\exists\x{Payment}
\end{align*}
After the initial research on those and other temporal description logics, recent studies have also considered metric operators and hence quantitative temporal logics (TLs)~\cite{AH-InfComp93}, 
yet mostly for more expressive logics \cite{GuJO-ECAI16:metric,BaBoKOT-FroCoS2017,BrandtKKRXZ-AAAI17}, such as for the description logic~(DL) \ALC, which is propositionally complete \cite{dlhandbook07}. 
Metric operators refer to concrete intervals and allow for describing temporal information more precisely. They are clearly also interesting for \DLLite, allowing, for instance, the statement that 
the payment must happen \emph{at some time in the future, maximally three} months after enrolling: 
$$\exists\x{EnrolledIn}\sqsubseteq {}\ieventf{[0,3]}\exists\x{Payment}$$

In temporal DLs, subsets of the symbols may additionally be distinguished as so-called \emph{rigid} symbols, to describe information that does not change over time. 
For instance, while a bachelor's student may become a master's student, he or she will always stay male or female. 
Declaring concepts $\x{Male}$ and $\x{Female}$ as rigid thus allows modeling the knowledge more faithfully, but often increases reasoning complexity.
Similarly, we may consider \emph{interval-rigid} symbols, to express that certain knowledge always must remain valid for a specific period of time \cite{BaBoKOT-FroCoS2017}. To describe that a master's degree lasts at least two years, the concept $\x{MasterStudent}$ may be declared as $24$-rigid.

\begin{table*}[t]
	\centering\renewcommand{\arraystretch}{1.1}
	\begin{tabular}{l@{\hspace*{1em}}l
			@{\hspace*{1em}}l@{\hspace*{0.5em}}r@{\hspace*{1em}}l@{\hspace*{0em}}r@{\hspace*{1em}}l@{\hspace*{0em}}}
		\toprule
		Special symbols
		&none
		& rigid
		&&interval-rigid
		&&interval-rigid,
		\\
		&
		&
		&& 
		&&only global CIs\\
		\midrule
		\LTLbinDLLbool
		&*
		& \TwoExpSpace& $\ge$\tabthm\ref{thm:ltlbin-dllbool}
		& \TwoExpSpace& $\le$\tabthm\ref{thm:alc-results}
		& \ExpSpace\hspace*{0.5em}$\le$\tabthm\ref{thm:alc-results}\\
		\LTLbinDLLkromhorn
		&*
		&\ExpSpace&$\le${\small [1]}
		&\ExpSpace&$\le$\tabthm\ref{thm:ltlbin-dllhorn-dllkrom}
		&\ExpSpace\\\midrule
		%
		LTL$_{\DLLite_{\textit{horn/bool}}}$
		&*
		&\ExpSpace& {\small [2]}
		&\ExpSpace& $\le$\tabthm\ref{thm:alc-results}
		&\PSpace\hfill $\le$\tabthm\ref{thm:ltl-dllbool-global}\\		
		LTL$_{\DLLitekrom}$
		&*
		&\PSpace&$\le${\small [2]}
		&\ExpSpace& $\ge$\tabthm\ref{thm:ltl-dllkrom-lb}
		&\PSpace\\\midrule
		$\DLLite_{\textit{bool}}$-\LTLbin
		&\ExpSpace 
		&\ExpSpace& 
		&\ExpSpace&$\le$\tabthm\ref{thm:alc-results}
		&\ExpSpace
		\\\midrule
		$\DLLite_{\textit{horn/bool}}$-LTL
		&\PSpace\hspace*{0.5em}$\le$\tabthm\ref{thm:dllbool-ltl-norigid-ub}
		&\NExpTime& \tabthm\ref{thm:dllhorn-ltl-rigid-lb}
		&\ExpSpace&$\ge$\tabthm\ref{thm:dllkromhorn-ltl-lb}
		&\PSpace\\
		\DLLitekrom-LTL
		&\PSpace
		&\PSpace&
		&\ExpSpace&$\ge$\tabthm\ref{thm:dllkromhorn-ltl-lb}
		&\PSpace\\
		\bottomrule
	\end{tabular}
	\caption{Complexity of satisfiability in \LTLbinDLLbool fragments w.r.t.\ rigid and interval-rigid names. In cases $*$, rigid symbols can be 
		simulated in the logic.
		The remaining \PSpace and \ExpSpace lower bounds follow from propositional 
		LTL \protect\cite{SiCl85:ltlpspace} and  MTL 
		\protect\cite{AH-JACM94:ltlbinexpspace}. 
		{\small [1] \protect\cite{GuJO-ECAI16:metric}\quad [2] \protect\cite{AKLWZ-TIME07:temporalising}}
	}
	\label{tab:interval-rigid-one}
\end{table*}
In this paper, we 
study the combined complexity of satis\-fiability in various fragments of $\LTLbinDLLbool$, the metric temporal extension of \DLLitebool (interpreted over integer time) where TL operators may be used both within the DL axioms and for combining them; the `$\mathsf{bin}$' hints at the binary encoding of interval boundaries. 
In a second dimension, we consider rigid and interval-rigid symbols.
Our complexity results are summarized in Table~\ref{tab:interval-rigid-one}. 
Next to the Bool fragment, we consider the Krom and Horn fragments of \DLLite, and thus extend the results of \cite{AKLWZ-TIME07:temporalising} on $\LTLDLLbool$ fragments regarding the metric operators and interval-rigid symbols. 
\LTLbinDLLbool is generally as expressive as \LTLDLLbool but exponentially more succinct.
Observe that satisfiability in \LTLbinDLLhorn is \ExpSpace-complete, the same as in \LTLDLLhorn \cite{AKLWZ-TIME07:temporalising}, but better than 
in \LTLbinALC, where we have \TwoExpSpace \cite{GuJO-ECAI16:metric,BaBoKOT-FroCoS2017}. 
%
Targeting better complexities, we also investigate the 
fragments where the TL operators must not occur within the DL axioms, denoted DL-TL (e.g., \DLLitebool-\LTLbin),
and, on the other hand, the restriction to \emph{global} CIs, which must not be combined arbitrarily 
but have to be always satisfied. 
%
%
Most importantly,
we show that the 
interval-rigid symbols cannot only be used to simulate rigid symbols, but also that their
expressive power 
may lead beyond the spirit of \DLLite: if the DL axioms can be arbitrarily combined with LTL operators, then interval-rigid symbols can be used to express the $\lnextf$-operator, and conjunction and disjunction 
in 
global CIs, so that 
many fragments collapse. 
%
%
We can show containment in \ExpSpace in most cases (note that \DLLitehorn-\LTLbin$\subseteq\LTLbinDLLhorn$ etc.). 
Moreover, we can extend the \PSpace result of \cite{AKRZ-TOCL14:cookbook} for \LTLDLLbool restricted to global CIs to the setting with interval-rigid names; note that this restriction does not seem to hurt in many applications, such as with conceptual modeling (see the introductory examples). Containment in \PSpace can also be shown for \DLLitebool-LTL, but then rigid symbols yield a surprising jump in complexity. Nevertheless, this result contrasts the \TwoExpTime-completeness we have for 
\ALC-LTL in this setting 
\cite{BaGL-TOCL12}. 
Our results strongly depend on the fact that, in some dialects, 
interval (and hence metric) operators must not occur within concepts, and on a result we show:
interval-rigid roles can be simulated through corresponding concepts.
This simplifies reasoning and shows that the \DLLite fragments represent rather special DLs
also in the quantitative temporal setting. 
All the results also hold if we consider \Nbb instead of \Zbb as the temporal dimension.

%
\section{Preliminaries}
\label{sec:preliminaries}

We first introduce \LTLbinDLLc and $\DLLitec$ for~$c\in\{\text{\textit{bool,horn,krom}}\}$ 
and establish the relation to \LTLbinALC.
%
%

\mypar{Syntax}
Let \NC, \NR and \NI be countably infinite sets of \emph{concept names},
\emph{role names}, and \emph{individual names}, respectively.
In $\DLLitec$, 
\emph{roles} and 
\emph{(basic) concepts} are defined 
as follows, where $P\in\NR,A\in\NC$: 
\begin{align*}
R &::=P \mid P^-,& 
B &:=\top\mid 
A \mid \exists R
\end{align*}
where $\cdot^-$ is the \emph{inverse role} constructor.

\emph{\LTLbinDLLbool concepts} are defined based on \DLLitebool concepts: 
$$C,D::= B\mid 
\lnext C \mid \lprevp C \mid C\iuntil D\mid C\isince D$$ where 
$I$ is an interval of the form 
$[c_1,c_2]$ or 
 $[c_1,\infty)$ with $c_1,c_2\in 
\Nbb$ and $c_1\leq c_2$,
given 
in \emph{binary}.
\emph{Concepts} in \LTLbinDLLhorn and \LTLbinDLLkrom
restrict \LTLbinDLLbool concepts in that they must not contain \emph{interval operators} ($\iuntil,\isince$).
%
%
%

\DLLitec \emph{axioms} are the following 
kinds of expressions:
\emph{concept inclusions} of the form
\begin{equation}
\label{eq:ci}
C_1\sqcap\dots\sqcap C_m\sqsubseteq C_{m+1}\sqcup\dots\sqcup C_{m+n}
\end{equation}
where 
$n\le1$ if $c=\horn$,  $m+n\le2$ if $c=\krom$, and $m,n\ge0$ if $c=\bool$;
and
\emph{assertions} of the form $C_1(a)$ or 
$R(a,b)$ with 
$C_1,\ldots,$ $C_{m+n}$ being \DLLitec concepts
and $a,b\in\NI$.

\LTLbinDLLc \emph{axioms} are defined accordingly but based on \LTLbinDLLc concepts $C_1,\ldots,$ $C_{m+n}$.

%

%
%
Define \LTLbinDLLc \emph{formulas} 
as follows:
%
$$\Phi,\Psi ::= \aaxiom\mid
\lnot\Phi\mid\Phi\land\Psi\mid\lnextf\Phi\mid\lprevp\Phi\mid
\Phi\iuntil\Psi\mid
\Phi\isince\Psi$$
where~$\aaxiom$ is an \LTLbinDLLc axiom, and $I$ is as above.

\DLLitec-MTL \emph{formulas} are defined accordingly, but based on  \DLLitec axioms \aaxiom.

\LTLDLLc restricts \LTLbinDLLc in that formulas must not contain \emph{quantitative} temporal operators ($\iuntil,\isince$), but they may contain the \emph{qualitative} versions $\luntil,\lsince$;
and correspondingly for \DLLitec-LTL and \DLLitec-MTL.

As usual, we denote the empty conjunction ($\sqcap$) by $\top$ and the empty 
disjunction ($\sqcup$) by $\bot$.
%
For a given MTL$_{\DLLitebool}$ formula \aform,
the sets of individual names, role names, and roles occurring in \aform are denoted, respectively, by $\individuals\aform$, $\NR(\aform)$, and $\NRM(\aform)$; and the closure under single negation of all concepts (formulas) and subconcepts (subformulas) occurring in \aform by $\csub(\aform)$ ($\fsub(\aform)$).
Note that $\csub(\aform)$ contains negated concepts, but we only consider these sets on a syntactic level so that it is not problematic that the semantics of such expressions is not defined in some dialects.
 $|\aform|$ is the size of \aform. We may also mention the Core fragment of \DLLite, the intersection of the Horn and the Krom fragment. 
In the following, we may use the notion \DLLite to address the fragments in general.


\mypar{Semantics}
A DL \emph{interpretation} $\Imc=(\Delta^\Imc,\cdot^{\Imc})$ over a non-empty 
set $\Delta^\Imc$, called the \emph{domain}, defines an \emph{interpretation 
function}~$\cdot^{\Imc}$ that maps each 
concept name $A \in \NC$ to a subset~$A^{\Imc}$ of~$\Delta^\Imc$, each role
name $P\in \NR$ to a binary relation~$P^\Imc$ on~$\Delta^\Imc$ and each 
individual name
$a \in \NI$ to an element~$a^{\Imc}$ of~$\Delta^\Imc$, such that $a^{\Imc}\neq 
b^{\Imc}$ if $a\neq b$, for all $a,b\in\NI$ ({unique name assumption}). 
%
The mapping~$\cdot^{\Imc}$ is extended 
to roles 
by defining $(P^-)^{\Imc} := {} \{(e,d) 
\in \Delta^\Imc \times \Delta^\Imc\mid
 (d,e) \in P^{\Imc}\}$ and to \DLLite concepts:
 $$\top^{\Imc} :=  {} \Delta^\Imc, \ 
 (\exists R)^{\Imc} := \{d \in \Delta^\Imc \mid \exists  e \in \Delta^\Imc : (d,e) \in R^{\Imc}\}.$$

A \emph{(temporal DL) interpretation} is a structure
$\Imf=(\Delta^\Imf,(\Imc_i)_{i\in\Zbb})$, where 
each
$\Imc_i=(\Delta^\Imf,\cdot^{\Imc_i})$, $i\in\Zbb$, is a DL
interpretation over~$\Delta^\Imf$ ({constant domain assumption}) and $a^{\Imc_i}=a^{\Imc_j}$ for all $a\in\NI$ and $i,j \in \Zbb$ (i.e., the
interpretation of individual names is fixed).
The mappings~$\cdot^{\Imc_i}$ are extended to \LTLbinDLLbool concepts as follows:
\begin{align*}
%
  (\lnextf C)^{\Imc_i} :={}& \{ d\in\Delta^\Imf \mid d\in C^{\Imc_{i+1}} \} \\
    (\lprevp C)^{\Imc_i} :={}& \{ d\in\Delta^\Imf \mid d\in C^{\Imc_{i-1}} \} \\
    (C\iuntil D)^{\Imc_i} :={}&\{ d\in\Delta^\Imf \mid
    \exists k\in i+I: d\in D^{\Imc_k}\ \text{and}\\& 
    \forall j\in \Zbb\colon \text{if }j\in i+[0,k), \text{ then }d\in C^{\Imc_{j}} \}\\
  (C\isince D)^{\Imc_i} :={}& \{ d\in\Delta^\Imf \mid
  \exists k\in i-I: d\in D^{\Imc_k}\ \text{and}\\& 
  \forall j\in \Zbb\colon \text{if }j\in (k,i],\text{ then }d\in C^{\Imc_{j}} \}
\end{align*}
%
where $i+I$ denotes the set $\{i+j \mid j\in I\}$ for all $i\in\Zbb$ and intervals $I$ as above; $i-I$ is defined analogously. 
%
The concept $C\iuntil D$ requires $D$ to be satisfied at some
point in the interval~$I$, and $C$ to hold at all time points before that; and similar for $\isince$.
The \emph{validity} of an \LTLbinDLLbool formula~$\aform$ in~\Imf at time point
$i\in\Zbb$ (written $\Imf,i\models\aform$) is inductively defined.
For CIs, we have $  \Imf,i\models C_1\sqcap \dots\sqcap C_m\sqsubseteq D_1\sqcup \dots\sqcup D_n  \text{ iff } C_1^{\Imc_i}\cap\dots\cap C_m \subseteq D_1^{\Imc_i}\cup \dots\cup D_n^{\Imc_i}
$; and further:
\[
\begin{array}{ll
	}
  \Imf,i\models C(a) & \text{ iff } a^{\Imc_i} \in C^{\Imc_i}\\
  \Imf,i\models R(a,b) & \text{ iff } (a^{\Imc_i},b^{\Imc_i}) \in R^{\Imc_i}\\
    \Imf,i\models \neg \aform & \text{ iff not } \Imf,i\models \aform\\
  \Imf,i\models \aform \wedge \aformtwo & \text{ iff } \Imf,i\models \aform
  \text{ and }\Imf,i\models \aformtwo \\
   \Imf,i\models \lnextf \aform & \text{ iff } \Imf,i+1\models \aform \\ 
   \Imf,i\models \lprevp \aform & \text{ iff } \Imf,i-1\models \aform \\ 
  \Imf,i\models \aform \iuntil\aformtwo
  & \text{ iff } \exists k\in i+I,\ \Imf,k\models \aformtwo, 
\text{ and }\\&\quad \forall j\in   \Zbb\colon \text{if }j\in i+[0,k), \text{ then } 
\Imf,j\models \aform\\
\Imf,i\models \aform \isince \aformtwo
& \text{ iff } \exists k\in i- I,\ \Imf,k\models \aformtwo, 
\text{ and }\\&\quad \forall j\in \Zbb\colon  \text{if }j\in (k,i],\text{ then }\Imf,j\models \aform
 \end{array}
\]
The Boolean operators $\bot,\lor,$ $\rightarrow,$ and $\leftrightarrow$ 
are defined as abbreviations in the
usual way. 
We further define 
$\alpha\luntil\beta := \alpha\iuntila{[0,\infty)}\beta$,
$\ieventf{I} \alpha := \top \iuntil \alpha$,
$\ialwaysf{I} \alpha := \neg (\top \iuntil \neg\alpha)$,
$\lalwaysf \alpha := \neg (\top \luntil \neg\alpha)$, 
$\lalwaysp \alpha := \neg (\top \isincea{(\infty,0]} \neg\alpha)$, and
$\boxast\alpha:=\lalwaysp\lalwaysf\alpha$, where $\alpha,\beta$ are either 
%
concepts or
formulas~\cite{dlhandbook07,GKWZ03:manydimmodal}. 
In accordance with the notation, the empty conjunction is interpreted as $\Delta^\Jmc$, and the empty disjunction as $\emptyset$. 
We may use negated concept names $\lnot A$ in \DLLitekrom, interpreted as $\Delta^\Imf\setminus A^\Imf$, which can be simulated by a fresh name $\overline{A}\in\NC$ via CIs $\top\sqsubseteq A\sqcup \overline{A}$ and $A\sqcap\overline{A}\sqsubseteq\bot$.
We further assume that neither $\lnextf$ nor $\lprevp$ occur in assertions, in the following; this is w.l.o.g., since we always allow the operators in front of assertions on axiom level.

\mypar{Relation to \LTLDLL}
\LTLDLLc restricts \LTLbinDLLc in that it only allows the qualitative temporal operators, but that 
%
%
does actually not decrease the expressivity:
every formula
$\aform\iuntila{[c_1,c_2]}\aformtwo$ can be transformed into an equisatisfiable formula
$\bigvee_{c_1\leq i\leq c_2}
  (\lnext^i \aformtwo \wedge \bigwedge_{0\leq j < i} \lnext^j \aform)$ and similarly for
concepts and for $\isince$; $\lnext^i$ denotes a sequence of $i$ $\lnext$-operators.
Likewise, $\aform\iuntila{[c_1,\infty)}\aformtwo$ is equivalent to
$\big( \bigwedge_{0\leq i<c_1} \lnext^i \aform \big) \land \lnext^{c_1} (\aform \luntil \aformtwo)$.
However, if this transformation is recursively applied to subformulas, then the size of
the resulting formula is exponential: ignoring the nested $\lnext$-operators,
its syntax tree has polynomial depth and an exponential branching factor; and the
$\lnext^i$-formulas have exponential depth, but introduce no branching.
This blowup cannot be avoided in
general~\cite{AH-InfComp93,GuJO-ECAI16:metric}.
Yet, an interesting result for \LTLziALC, the restriction of \LTLbinALC to intervals of the form $[0,c]$ and
$[c,\infty)$, is given in \cite[Thm.~2]{BaBoKOT-FroCoS2017}, 
and that reduction works similarly in our setting with past operators: each \LTLziDLLc formula can be translated in polynomial time into an
equisatisfiable \LTLDLLc formula.
%
The reduction is particularly modular in that, if the formula contains  only global
CIs (which are formally introduced in the next paragraph), 
then this is still the case after the reduction. 

\mypar{Reasoning}
We study the complexity of the \emph{satisfiability} problem in
\LTLbinDLLbool (and in its fragments): given an \LTLbinDLLbool formula~\aform, decide if there exists an interpretation~\Imf such that
$\Imf,0\models\aform$.
Additionally, we consider a syntactic restriction proposed in~\cite{BaGL-TOCL12}.
An \LTLbinDLLbool formula \aform is a  \emph{formula with global CIs} if it is of the form
$\boxast\Tmc\wedge\aformtwo$, where \Tmc is a conjunction of CIs and \aformtwo
is an \LTLbinDLLbool formula that does not contain CIs.
\emph{Satisfiability w.r.t.\ global CIs} represents the satisfiability
problem 
w.r.t.\ such formulas. The problems and notions are correspondingly defined for the fragments of \LTLbinDLLbool.

\mypar{Rigid Names}
We especially consider a finite set $\NRig\subseteq\NC\cup\NR$ of \emph{rigid} symbols, whose interpretation must not change over time. That is,  interpretations
$\Imf=(\Delta^\Imf,(\Imc_i)_{i\in\Zbb})$ must \emph{respect} these names, meaning: $X^{\Imc_i}=X^{\Imc_j}$ for all $X\in\NRig$ and
$i,j\in\Zbb$.
%
%
%
In addition, we consider a finite set
$\NIRig\subseteq(\NC\cup\NR)\setminus\NRig$ of \emph{interval-rigid} names, each of which must remain rigid for a specific period of time, determined by a function $\ir\colon\NIRig\to\Nbb_{\ge 2}$ whose values are given in binary.
Interpretations $\Imf=(\Delta^\Imf,(\Imc_i)_{i \in \Zbb})$ must also
\emph{respect} these names, meaning: for all
$X\in\NIRig\cap\NC$ with $\ir(X)=k$ and $i\in\Zbb$: for every $d\in X^{\Imc_i}$, there is a time point $j \in \Zbb$ such that
	$i\in[j,j+k)$ and $d\in X^{\Imc_\ell}$ for all
	$\ell\in[j,j+k)$; and similarly for role names.
%
%
In the following, let 
$\const:=\max \{ \ir(X)\mid X\in\NIRig \text{ occurs in }\aform \}$. 
Intuitively, any element (or pair of elements) in the interpretation of an
interval-rigid name must be in that interpretation for at least~$k$ consecutive
time points; the name is \emph{$k$-rigid}.
The names in $(\NC\cup\NR)\setminus(\NRig\cup\NIRig)$ are
\emph{flexible}.
%
%
We investigate the complexity of satisfiability w.r.t.\ different settings, in dependence of which kinds of (interval-)rigid names may occur in the formula.

\cite{AKLWZ-TIME07:temporalising} mention that rigid roles can be simulated using temporal
constraints on unary predicates. In fact, in an \LTLDLLc formula \aform, a rigid role name $R$ can be simulated by considering $R$ to be flexible, introducing fresh 
rigid concept names $C_{\exists R}$ and $C_{\exists R^-}$, 
and extending \aform with the conjunct $ \boxast((C_{\exists R}\sqsubseteq\exists R)\land (C_{\exists R^-}\sqsubseteq\exists R^-)\land (\exists R\sqsubseteq C_{\exists R})\land (\exists R^-\sqsubseteq C_{\exists R^-}))$, and with a conjunct $(\lnot R(a,b)\lor\boxast R(a,b))$ for each role assertion $R(a,b)$ occurring in \aform. Note that this reduction even works in the Core fragment, does not require temporal operators on the DL level, and only uses global CIs. We can extend the reduction to interval rigid symbols: for every $k$-rigid role $S$, we introduce fresh $k$-rigid concept names $C_{\exists S}$ and $C_{\exists S^-}$ and CIs corresponding to the above ones, and we extend \aform
with the conjunct
$
\boxast\left( \lnot S(a,b)\to\lnext(\lnot S(a,b)\lor \ialways{[0,k)}S(a,b)) \right)$ 
for each role assertion $S(a,b)$ occurring in \aform.
An \LTLbinDLLc formula referring only to intervals of the form $[0,k)$ is in the \LTLziDLLc fragment and can be translated in polynomial time into an equisatisfiable
$\LTLDLLc$ formula, as mentioned above; note that
the interval operator does not occur on the concept level.
%
\begin{lemma}\label{lem:rig-roles-trivial}
	Satisfiability in 
	\LTLbinDLLc 
	w.r.t.\ (interval-)rigid names can be polynomially reduced to the setting
	where only (interval-)rigid concepts are given.
	%
	\qed
\end{lemma}
We hence restrict our attention to (interval-)\ rigid \emph{concepts}. 
Since \LTLDLL allows to express rigid concepts axiomatically using CIs $A\sqsubseteq\lnext A$ and $\lnext A\sqsubseteq A$, rigid symbols are actually 
syntactic sugar in that language. This does not seem to be the case for fragments that do not allow for temporal operators within CIs, which we consider later, but, at least for the case with interval-rigid names, we will prove the contrary. 
Moreover, $k$-rigid concept names $A$ can be simulated using 
$\top\sqsubseteq A\sqcup\lnext (\lnot A\sqcup\ialways{[0,k)} A)$, 
and hence in \LTLDLLbool. 

\mypar{Relation to \LTLbinALC}
\ALC is more expressive than \DLLitebool although it does not allow for inverse roles. Also,
every \LTLbinDLLbool 
formula \aform can be transformed into an equisatisfiable \LTLbinALC formula $\aform'$ as follows. 
First, we extend the set of role names to include all inverse roles $R^-$ for which we have that the role
$R$ or $R^-$ occurs in \aform. 
Then, $\aform'$ is obtained from \aform by replacing
all occurrences of concepts of the form $\exists R$ with $R\in\NRM$ by
$\exists R.\top$%
\footnote{%
	In an interpretation~$\aint=(\Delta^\aint,\cdot^\aint)$, a concept
	$\exists R.C$ is interpreted as the set $\{d\in\Delta^\aint \mid\exists(d,e)\in R^\aint: e\in C^\aint\}$.	}; by adding a CI $\boxast(\exists R.(\lnot\exists R^-.\top)\sqsubseteq\bot)$ for each
$R\in\NRM(\aform)$ as a conjunct;
and by adding a conjunct $\boxast(R(a,b)\leftrightarrow R^-(b,a))$ for each role assertion $R(a,b)$ or $R^-(b,a)$ occurring in \aform.
%
Note that 
we use no 
temporal operators within the CIs, no metric ones, and that the CIs are global. 
\begin{restatable}{lemma}{LemRedToALC}\label{lem:red-to-alc}
	The \LTLbinDLLbool 
	formula \aform is satisfiable iff the \LTLbinALC 
	formula
	$\aform'$ is satisfiable.
\end{restatable} 
\begin{proof}
	($\Leftarrow$)
	We consider an arbitrary model $\Imf=(\Delta,(\Imc_i)_{i\in\Zbb})$ of \aform and construct a model $\Imf'=(\Delta,(\Imc'_i)_{i\in\Zbb})$ of $\aform'$. Specifically, $\Imf'$ interprets all symbols occurring in \aform as \Imf does, and, for all $i\in[0,n]$,
	the interpretation of role names $R^-$ in $\Imc'_i$, where
	$R^-\in\NR(\aform')\setminus\NR(\aform)$, is equal to~$(R^-)^{\Imc_i}$.
	Given this definition of $\Imf'$, we obviously have that, for all $i\in\Zbb$, 
	$\Imc'_i$ satisfies an axiom \aaxiom occurring in \aform iff $\Imc_i\models\aaxiom$. Moreover, it is easy to see that the new axioms are satisfied, too. We 
	thus have $\Imf'\models\aform'$.
	
	($\Rightarrow$)
	Let now $\Imf'=(\Delta,(\Imc'_i)_{i\in \Zbb})$ be a model of $\aform'$.
	We show this direction similarly, by constructing a model $\Imf=(\Delta,(\Imc_i)_{i\in\Zbb})$ of \aform.
	In particular, we assume \Imf to have the same domain as $\Imf'$, to interpret 
	all concept names as $\Imf'$ does, and to interpret all role names 
	$R\in\NR(\aform)$ such that $R^{\Imc_i}=R^{\Imc'_i}\cup\{ (\elone,\eltwo) \mid 
	(\eltwo,\elone)\in  (R^-)^{\Imc'_i}\}$ for all $i\in\Zbb$. 
	The latter definition yields that $(\elone,\eltwo)\in R^{\Imc_i}$ if 
	$(\elone,\eltwo)\in R^{\Imc'_i}$, and $(\elone,\eltwo)\in (R^-)^{\Imc_i}$ if $(\elone,\eltwo)\in 
	(R^-)^{\Imc'_i}$ for all $i\in\Zbb$.
	Together with (i) and the definition of $R^{\Imc_i}$, we thus obtain that 
	$e\in (\exists R)^{\Imc_i}$ iff $e\in (\exists R)^{\Imc'_i}$ for 
	all $R\in\NRM(\aform)$.
	Given the latter and the fact that $\Imf'$ satisfies our extension of \aform regarding the role assertions, it is easy to see that we get $\Imc_i\models\aaxiom$ iff $\Imc'_i\models\aaxiom$ for all assertions \aaxiom occurring in \aform.
	We get the same for CIs \aaxiom.
	Hence, we have $\Imf\models\aform$.			
\end{proof}
%
The reduction yields the below membership results. 
%
%
\begin{restatable}{theorem}{ThmALCResults}
	\label{thm:alc-results}
	Satisfiability w.r.t.\ both rigid and interval-rigid names is in \ExpSpace in the following logics:
	\begin{itemize}
		\item \LTLbinDLLbool restricted to global CIs,
		\item \LTLDLLbool,
		\item \DLLitebool-\LTLbin.
	\end{itemize}
In \LTLbinDLLbool, the problem is in \TwoExpSpace.
\end{restatable}
\begin{proof}
	By Lemma~\ref{lem:rig-roles-trivial}, we can refer to the complexities for \ALC w.r.t.\ rigid and interval-rigid concepts only to obtain the results. 
	The corresponding results all have been shown in \cite{BaBoKOT-FroCoS2017},
	for \LTLbinALC restricted to global CIs, 
	\LTLALC, 
	\ALC-\LTLbin, 
	and for \LTLbinALC, 
	but only w.r.t.\ the fragments of the logics with only future operators and the natural numbers. 
	
	All the proofs apply the standard approach. In what follows we sketch it regarding the \ExpSpace cases; the other one is similar. It is based on the fact we can restrict the focus to a kind of 
	models---so-called quasimodels---of a special, regular shape and, specifically, to a part of double-exponential size (one version of such a proof is presented in the next section). 
	Such a quasimodel is a sequence of quasiworlds. A quasiworld describes the interpretation of all domain elements at a single time point (through closed subsets of $\csub(\aform)$) and contains a closed set of subformulas of $\aform$, those that are satisfied at the corresponding time point. Additionally, specific conditions hold for consecutive quasiworlds in a quasimodel, to guarantee that it describes the interpretations of all domain elements that have to be present to satisfy the given formula at $0$ w.r.t.\ all time points.
	In such a setting,
%
	we can restrict our focus to quasimodels of the form $^\omega Q_0 Q_1  Q_2^{\omega}$, where $Q_0,Q_1$, and $Q_2$ are sequences of quasiworlds of double-exponential length, $Q_0$ and $Q_2$ do not contain a quasiworld twice, and $Q_1$ does not contain a quasiworld more than twice. 
	In a nutshell, this is shown by merging models of the form $\dots Q'_1* w*\cdots *w*Q''_1  Q_2 Q_3$, containing a quasiworld $w$ twice, to models $\dots Q'_1* w*Q''_1Q_2 Q_3$ with $Q'_1* w*Q''_1=Q_1$. $Q_2$ is obtained from an arbitrary given sequence of quasiworlds $Q_2 Q_3$ by considering all subconcepts containing the $\luntil$-operator---w.r.t.\ all domain elements described in the model---, and also the subformulas containing $\luntil$. It can be shown that there is always a sequence as $Q_2$ (i.e., one of the length of $Q_2$), in which all these concepts are finally satisfied by the respective elements, and similar for the formulas and in the past direction; and that $^{\omega}Q_0 Q_1  Q_2^{\omega}$ describes a model of \aform. 
	%
	We regard the \ExpSpace cases.
	In the original proofs, and also in our setting, the number of different quasiworlds is bounded by
	$\sharp(\aform)=2^{2^{|\csub(\aform)|}}*|\individuals\aform|*2^{|\csub(\aform)|}*2^{|\fsub(\aform)|}$, meaning
	double exponentially in the input.
	Hence this also holds for the lengths $|Q_1|\le2\sharp(\aform)$ of $Q_1$ and $|Q_2|\le \flat^2* |\csub(\aform)|*\sharp(\aform)+|\fsub(\aform)|*\sharp(\aform)+\sharp(\aform)$ with $\flat=2^{|\csub(\aform)|}+|
	\individuals\aform|$, and for $|Q_0|$. 
	Therefore, the existence of a model of \aform can be checked while using only exponential space: 
	first, guess the starts of the periods~$n,n'\le\sharp(\aform)$ and their lengths $m,m'\le\flat^2* |\csub(\aform)|*\sharp(\aform)+|\fsub(\aform)|*\sharp(\aform)+\sharp(\aform)$ with $\flat=2^{|\csub(\aform)|}+|
	\individuals\aform|$; second, guess the two sequences of quasiworlds $w(-1),\ldots,w(-(n+m))$ (written in the order of the guessing) and $w(0),\ldots,w(n'+m')$ by guessing one 
	world after the other, respectively. Thereby only three quasiworlds have to 
	be kept in memory at a time---the ``current'' quasiworld, the previous (next) one, 
	and the first repeating one $w(n+1)$ ($w(n'+1)$)---and their sizes are exponentially 
	bounded in the size of the input.

	The proof for the \TwoExpSpace case is similar, but the models are more complex and the bound there is triple exponential.
\end{proof}
Compared to other description logics, the \DLLite logics thus present rather special cases
also in the temporal setting. This is mainly due to the facts that rigid roles can be disregarded, and that, in some dialects, 
the interval operators may not occur within concepts. It is not directly clear how this affects the complexity results. 
In the remainder of the paper, we therefore look for fragments of \LTLbinDLLbool where satisfiability is not as complex as in the corresponding \LTLbinALC fragment.

\section{\LTLbinDLL and Interval-Rigid Names}
\label{sec:temp-ops-on-all}
We begin focusing on \LTLbinDLLbool and the fragments where temporal operators may occur on both concept and axiom level.
Recall that we trivially have rigid symbols in these logics 
(Lem.~\ref{lem:rig-roles-trivial}).
Alas, the reduction of the word problem of double-exponentially space-bounded deterministic Turing machines from the \TwoExpSpace-hardness proof 
for \LTLbinALC over the natural numbers \cite[Thm.\ 5]{GuJO-ECAI16:metric} can be similarly done in \LTLbinDLLbool.
\begin{restatable}{theorem}{ThmLtlbindllbool}\label{thm:ltlbin-dllbool}
	Satisfiability in \LTLbinDLLbool without interval-rigid names is \TwoExpSpace-hard. 
\end{restatable}
\begin{proof}
	The proof proposed for \LTLbinALC 
	is based on a reduction of the word problem of double-exponentially
	space-bounded deterministic Turing machines. The \LTLbinALC formula in that proof contains qualified existential restrictions on the right of CIs, sometimes prefixed by $\levent$, but not nested and, apart from that, only constructs that are allowed in \LTLbinDLLbool. In particular, all the qualified existential restrictions are of the form $\exists R.\dots$, meaning that they all use the same role $R$. Moreover, it can readily be checked that this feature is not critical since the role is otherwise not used in the formula. That is, the hardness result depends on the element the existential restriction forces to exist but not on the kind of the relation to its predecessor. Consequently, for each such restriction $\exists R.C$, we can introduce a fresh role name $\exists R_{C}$, and then create a similar \LTLbinDLLbool formula by adding the conjuncts $\lalways(\exists R_{C}^-\sqsubseteq C)$ for the introduced role names to the \LTLbinALC formula and by replacing all concepts $\exists R.{C}$ by~$\exists R_{C}$.
\end{proof}%
Given the \TwoExpSpace-hardness, the results in Thm.~\ref{thm:alc-results} for restrictions of \LTLbinDLLbool, or  \LTLbinALC, 
are really interesting. 

Alternatively, we can regard other \DLLite fragments. 
In fact, we show that satisfiability w.r.t.\ interval-rigid names is in \ExpSpace in both \LTLbinDLLhorn and \LTLbinDLLkrom.
This is particularly the case because 
$\iuntil$ and $\isince$ do not occur in concepts there. 
%
Our proof is an extension of the one for \LTLALC \cite{WoZ-FroCoS90:temporalizing}
   regarding interval-rigid names. 
%
%
Assume $\aform'$ to be the given formula, and \aform to be the exponentially larger formula obtained from it by simulating all $\iuntil$- and $\isince$-operators (see Sec.~2). Note that $\csub(\aform)=\csub(\aform')$ since $\iuntil$ and $\isince$ here do not occur in concepts, but $\fsub(\aform)$ is exponentially larger than $ \fsub(\aform')$.

A \emph{concept type for \aform} is a set $t$ 
as follows: 
%
\begin{enumerate}[label=\textbf{T\arabic*},leftmargin=*]
	\item $\lnot\top\not\in t$;
	\item\label{ct:neg} $\neg C\in t$ iff $C\not\in t$, for all
	$\neg C\in\cons\aform$; 
	\item\label{ct:ir} $A\in t$ iff $A^i\in t$ for exactly one $i\in[1,\ir(A)]$,
	for all $A\in\cons\aform\cap\NIRig$;
	\item\label{ct:ir2} 
	%
	$A^i\in t$ for some $i\in[1,\ir(A))$ implies $\lnext A\in t$, 
	for all $A\in\cons\aform\cap\NIRig$;
	\item\label{ct:ir3} 
	%
	$A^i\in t$ for some $i\in(1,\ir(A)]$ implies $\lprevp A\in t$, 
	for all $A\in\cons\aform\cap\NIRig$;
\end{enumerate}
the names $A^i$ are used to capture how long $A$ has been satisfied already.
A \emph{named} concept type for \aform is a pair $(a,t)$ with $a\in\individuals\aform$ and a concept type $t$ for \aform.
We denote such a tuple by $t_a$ and write $C\in t_a$ instead of
$C\in t$.
\emph{Formula types} $t\subseteq\fsub(\aform)$ are defined by the
following conditions:
\begin{enumerate}[label=\textbf{T\arabic*'},leftmargin=*]
	\item\label{ft:neg} for all
	$\neg\alpha\in\fsub(\aform)$, $\neg\alpha\in t$ iff $\alpha\not\in t$; 
	\item\label{ft:fcon} for all
	$\alpha\wedge\beta\in\fsub(\aform)$, $\alpha\wedge\beta\in t$ iff $\alpha,\beta\in t$.
\end{enumerate}
Intuitively, a concept type describes the interpretation w.r.t.\ one domain element at a single time point; a formula type specifies constraints on the whole domain. 

A \emph{quasiworld} for~\aform is a triple $w=(\Tmc,\Tmc\ind,\Fmc)$, where $\Tmc$ is a set of unnamed types, 
$\Tmc\ind$ is a set of named types containing exactly one named type for each $a\in\individuals\aform$, $\Fmc$ is a formula type, and:
\begin{enumerate}[label=\textbf{W\arabic*},leftmargin=*]
	\item\label{q:gci} 
	for all $\aaxiom=(C_1\sqcap\dots\sqcap C_m\sqsubseteq C'_1\sqcup\dots\sqcup C'_n)\in\fsub(\aform)$, 
	$\aaxiom\in \Fmc$ iff $C_1,\dots, C_m\in t$ implies $\{C'_1,\dots,C'_n\}\cap t\neq\emptyset$ for all types $t\in \Tmc\cup \Tmc\ind$;
	\item\label{q:ex} for all $t\in \Tmc\cup \Tmc\ind$ and $R\in\NRM$, 
	$\exists R\in t$ iff there is a $t'\in \Tmc$ 
	with $\exists R^-\in t'$. 
	%
	\item\label{q:cass} for all $C(a)\in\fsub(\aform)$, 
	$C(a)\in \Fmc$ iff $C\in t_a$; 
	\item\label{q:rass} for all $R(a,b)\in\fsub(\aform)$, 
	$R(a,b)\in \Fmc$ implies $\exists R\in t_a$ and $\exists R^-\in t_b$.
\end{enumerate}
Note that $|\Tmc|\le 2^{|\csub(\aform)|}* \const$, $|\Tmc\ind|=|\individuals\aform|$, $|\Fmc|\le 2^{|\fsub(\aform)|}$, 
and that the number of distinct quasiworlds for \aform is double exponential and does not exceed
$$\sharp = 2^{2^{|\csub(\aform)|}* \const} *
|\individuals\aform| * 2^{|\csub(\aform)|}* \const *
2^{|\fsub(\aform)|}.
$$
A quasiworld describes an interpretation at one time point.
Regarding several time points, we first consider single sequences of types, which describe an interpretation on one element w.r.t.\ all time points.

A pair $(t,t')$ of concept types is \emph{suitable} if we have:
\begin{itemize}
	\item
	for all $\lnext C \in \cons\aform$, $\lnext C\in t$ iff $C\in t'$;
	\item
	for all $\lprevp C \in \cons\aform$, $C\in t$ iff $\lprevp C\in t'$;
	\item
	for all $A\in\cons\aform\cap\NIRig$ and $j\in[1,\ir(A)]$, $A^j\in t$ iff either $j=\ir(A)$ and  $\lnot A\in t'$, $A^j\in t'$, or $A^{j+1}\in~t'$.
\end{itemize}
Let $\mathfrak{w} =(\dots ,w_{-1},w_0,w_1,\dots)$ ($*$) be a sequence of quasiworlds $w_i=(\Tmc_i,\Tmc\indi,\Fmc_i)$ for \aform.
We denote concept types in $\Tmc\indi$ by $t^i_a$ for $a\in\individuals\aform$ and, w.l.o.g., 
assume that every concept type in $\Tmc\ind$ also occurs in \Tmc.
%
A \emph{run} in $\mathfrak{w}$ is a sequence $r$
of concept types such that, for all $i\in\Zbb$:
\begin{enumerate}[label=\textbf{R\arabic*},leftmargin=*,series=run]
	\item\label{enu:run0} $r(i)\in \Tmc_i$; 
	\item\label{enu:run1} the pair $(r(i),r({i+1}))$ is suitable;
	
	
\end{enumerate}
$r(i)$ denotes the element at index $i$ in a sequence~$r$.

Finally, a sequence $\mathfrak{w}$ of the form ($*$) is a \emph{quasimodel} for~\aform if the following hold for all $i\in\Zbb$:
\begin{enumerate}[label=\textbf{\upshape{M\arabic*}},leftmargin=*]
	\item\label{enu:m1} for all $a\in\individuals\aform$, 
	$(\dots,t_a^{-1},t_a^{0},t_a^{1},\dots)$ 
	is a run in~$\mathfrak{w}$;
	\item\label{enu:m2} for all 
	$t\in \Tmc_i$, there is a run $r$ in $\mathfrak{w}$ with $r(i)=t$;
	\item\label{enu:m3} for all 
	$\lnext\alpha\in\fsub(\aform)$, 
	$\lnext\alpha\in \Fmc_i$ iff $\alpha\in \Fmc_{i+1}$;
	\item\label{enu:m32} for all 
	$\lprevp\alpha\in\fsub(\aform)$, 
	$\alpha\in \Fmc_i$ iff $\lprevp\alpha\in \Fmc_{i-1}$;
		\item\label{enu:m4} for all $\alpha\iuntil\beta\in\fsub(\aform)$, 
		$\alpha\iuntil\beta\in \Fmc_i$ iff there is a
		 $k\in i+I,\ \Imf,k\models \beta, 
		\text{ and, for all }j\in   \Zbb\colon \text{if }j\in i+[0,k), \text{ then } 
		\Imf,j\models \alpha$; 
		\item\label{enu:m42} for all $\alpha\isince\beta\in\fsub(\aform)$, 
		$\alpha\isince\beta\in \Fmc_i$ iff there is a
		 $k\in i- I,\ \Imf,k\models \beta, 
				\text{ and, for all } j\in \Zbb\colon  \text{if }j\in (k,i],\text{ then }\Imf,j\models \alpha$;
	
	\item\label{enu:m5}\label{m:phi} 
	$\aform\in \Fmc_0$.
\end{enumerate}
\begin{restatable}{lemma}{Lemaux}
	An \LTLbinDLLkromhorn
	formula \aform is satisfiable w.r.t.\ interval-rigid names
	iff there is a quasimodel for \aform.
\end{restatable}
\begin{proof}
	($\Rightarrow$) Suppose that \aform is satisfied in an interpretation $\atint=(\Delta,(\aint_i)_{i\in\Zbb})$ that respects the interval-rigid names. For every $i\in\Zbb$, we define the quasiworld $w_i=(\Tmc_i,\Tmc\indi,\Fmc_i)$ as follows:
	\begin{align*}
	\Tmc_i:={}&\{  t_{\Imc_i}(e)\mid e\in\Delta  \},\\
	\Tmc\indi:={}&\{  (a,t_{\Imc_i}(a))\mid a\in\individuals\aform \},\\
	\Fmc_i:={}&\{ \alpha\in\fsub(\aform)\mid\Imc_i\models\alpha \},
	\end{align*}
	where $t_{\Imc_i}(e):= \{ C\in\csub(\aform)\mid  \Imc_i\models C(e) \}
	\cup \{ A^\ell\in\csub(\aform)\cap\NIRig\mid
	\exists j \le i:\forall k>j:\ \text{if}\ k\le i,\ \text{then } \Imc_k\models A(e)\ \text{and}\ \ell=\min(i-j+1,\ir(A)) \}$.
	Clearly, every $w_i$ represents a quasiworld, and 
	the sequence $(\dots ,w_{-1},w_0,w_1,\dots)$ is a quasimodel for \aform.
	
	($\Rightarrow$) Let a quasimodel $\mathfrak{w}$ for \aform of the form ($*$) be given. We  define the interpretation
	$\Imf = (\Delta^\Imf, (\Imc_i)_{i\in\Zbb})$ as follows, based on the set \Rmf of 
	runs in $\mathfrak{w}$:
	\begin{align*}
	\Delta^\Imf :={}&\{d_r \mid r \in\Rmf\} \cup\individuals\aform\\
	a^{\Imc_i}  :={}& a\\
	A^{\Imc_i}  :={}& \{d_r\mid A \in r(i),\ r \in \Rmf\}\cup \{a\mid A \in t^i_a\}  \\
	R^{\Imc_i}  :={}& \{ (d_r,d_{r'}) \mid r,r'\in\Rmf, \exists R\in r(i), \exists R^-\in r'(i) \} \cup{}\\& \{ (a,b)\mid R(a,b)\in\fsub(\aform) \}\cup{}\\& \{ (a,d_{r}) \mid r\in\Rmf, \exists R^-\in r(i), \exists R\in t^i_a\}\cup{}\\& \{ (d_{r},a) \mid r\in\Rmf, \exists R\in r(i), \exists R^-\in t^i_a\}.
	\end{align*}
	Below, we sometimes denote $a\in\individuals\aform$ by $d_r$, where $r$ is the unique run for $a$ in $\mathfrak{w}$.
	%
	Given \ref{enu:run1}, we directly have that \Imf respects interval-rigid concepts.
	
	To show that \Imf is a model of \aform, we prove the following claim.
	Note that, by our assumption that $\Tmc\indi\subseteq\Tmc_i$ for all $i\in\Zbb$ and by \ref{enu:m1}, it also covers the named elements.
	%
	\begin{claim*}
		For all runs $r\in\Rmf$, concepts $C\in\csub(\aform)$ and $i\in\Zbb$, we have
		$ C \in r(i)$ iff $ d_r \in C^{\Imc_i}.$
	\end{claim*}
	\smallskip
	\noindent{\it Proof of the claim.}
	We argue by structural induction. Clearly, the claim holds for all
	concept names. 
	It thus remains to consider the operators $\exists $, $\lnext$, and $\lprevp$.
	
	Let $C=\exists R$: If $\exists R \in r(i)$, then
	by~\ref{q:ex} and \ref{enu:m2} there is an unnamed run $r'$ such that $\exists R^-\in r'(i)$.
	By the definition of~$R^{\Imc_i}$, we get $(d_r,d_{r'})\in R^{\Imc_i}$, and hence
	$d_r\in \exists R^{\Imc_i}$.
	Conversely, $d_r\in\exists R^{\Imc_i}$ directly yields $\exists R\in r(i)$ by the definition of $R^{\Imc_i}$.
	
	Let $C=\lnext D$: We have that $\lnext D \in r(i)$ iff
	$D \in r(i+1)$ (by \ref{enu:run1}) iff $d_r\in D^{\Imc_{i+1}}$ (by
	induction) iff $d_r\in(\lnext D)^{\Imc_i}$ (by the semantics).
	The proof for $\lprevp$ is analogous.
	%
	%
	Using \ref{q:gci}, \ref{q:cass}, \ref{q:rass}, \ref{enu:m3}--\ref{enu:m42}, and similar arguments as above for concepts, we can now
	show that, for all $i\in\Zbb$ and $\aformtwo\in\fsub(\aform)$, it holds that $\Imf,i\models\aformtwo$
	iff $\aformtwo\in \Fmc_i$.
	Hence, by~\ref{enu:m5}, we get that $\Imf,0\models\aform$.
	%
\end{proof}
Observe that we extend the original proof from \cite{WoZ-FroCoS90:temporalizing} only in that we consider \Zbb instead of \Nbb, intervals with the operators $\luntil$ and $\lsince$, and interval-rigid concepts;
especially the former extensions are irrelevant.
%
	It is hence possible to consider only quasimodels of the form $^{\omega}Q_0 Q_1 *w_0*Q_2 Q_3^\omega$,\footnote{For brevity, we may drop the brackets around 
		sequences.} where $Q_0,Q_1,Q_2,$ and $Q_3$ are sequences of quasiworlds of double-exponential length 
	not containing a quasiworld twice, as outlined above.
Recall that the  number $\sharp(\aform)$ of different quasiworlds is bounded double exponentially in the input $\aform'$. This yields the following result.
\begin{lemma}
	If \aform has a quasimodel, then it has a quasimodel 
	of the form
	$^{\omega}(w_{-(n+m)}\dots w_{-(n+1)}) w_{-n}\dots w_0\dots $ $ w_{n'}( w_{n'+1}\dots w_{n'+m'})^{\omega}
	$
	such that $n,n'\le\sharp(\aform)$ and $m,m'\le|\fsub(\aform)|*\sharp(\aform)+\sharp(\aform)$ are bounded double exponentially in the size
	of~$\aform'$ and~\ir. \qed
\end{lemma}
The existence of a model of $\aform'$ can thus again be checked by guessing the starts and lengths of the periods and the worlds, one after each other, while using only exponential space.
\begin{theorem}\label{thm:ltlbin-dllhorn}\label{thm:ltlbin-dllkrom}\label{thm:ltlbin-dllhorn-dllkrom}
	Satisfiability w.r.t.\ interval-rigid names is in \ExpSpace in
	\LTLbinDLLkromhorn. \qed
\end{theorem}

\begin{figure*}[t]\centering\small
	\begin{tabular}{llllllll}
		\toprule
		$j$&\dots&-1&0&1&2&3\hfill\dots\\\midrule
		$A_0^{\Imc'_j}$ && $(C\sqcap D)^{\Imc_{0}}$& 
		$ C^{\Imc_0}$& 
		$(C\sqcap \lnot D)^{\Imc_{0}}$&
		$(C\sqcap  D)^{\Imc_3}$& 
		$ C^{\Imc_3}$\\
		$A_1^{\Imc'_j}$&&$(C\sqcap \lnot D)^{\Imc_{-2}}$& 
		$(C\sqcap D)^{\Imc_1}$& 
		$C^{\Imc_1}$&$(C\sqcap \lnot D)^{\Imc_{1}}$& $(C\sqcap D)^{\Imc_4}$& 
		\\
		$A_2^{\Imc'_j}$&& $ C^{\Imc_{-1}}$&$(C\sqcap \lnot D)^{\Imc_{-1}}$& 
		$(C\sqcap D)^{\Imc_{2}}$&
		$ C^{\Imc_2}$&$(C\sqcap \lnot D)^{\Imc_{2}}$ &	
		\\
		%
		\bottomrule
	\end{tabular}
	\caption{Outline of the interpretation of the $2$-rigid concept names $A_i$ for $i\in[0,2]$ in $\Imf'$ at time $j$, in the proof of Thm.~\ref{thm:ltl-dllkrom-lb}. }
	\label{fig:ex}
\end{figure*}
Regarding interval-rigid names, Thm.~\ref{thm:alc-results} states the best result possible for \LTLDLLbool and \LTLDLLhorn, but satisfiability in \LTLDLLkrom is in \PSpace \cite[\citethm5]{AKLWZ-TIME07:temporalising}\footnote{\cite[\citethm5]{AKLWZ-TIME07:temporalising} refer to formulas with only future operators and the natural numbers, but the proof can be extended to our version of \LTLDLLkrom.}. 
Alas, interval-rigid names destroy this result:
we can express global CIs of the form $C\sqcap D\sqsubseteq E$ ($*$) 
---the main ${\DLLitehorn}$ feature%
\footnote{It is well known that CIs of the form \eqref{eq:ci} with $n=1$ can be simulated by CIs of the same form where $m=2$.}---using the following formula $\aformtwo$: \[\boxast\Bigg(\bigvee_{0\le i \le 2} \Psi_i\Bigg)
\land\boxast\Bigg(\bigwedge_{0\le i \le 2} \Psi_i\to\lnext\Psi_{i\oplus_31}\Bigg)\]
$\Psi_i$  is the conjunction of the following CIs: 
\begin{align*}
C\sqsubseteq{}& A_i&
D\sqsubseteq{}&\lnot \lnext A_i&
\lprevp A_{i}\sqsubseteq{}&  E 
\end{align*}
with 
$A^j$, $j\in[0,2]$, being fresh, $2$-rigid concept names; we use 
${\oplus_i{}}$ ($\ominus_i{}$) to denote addition (subtraction) modulo~$i$. 
Then, every model $\Imf=(\Delta,(\Imc_i)_{i\in \Zbb})$ of ($*$) can be extended to a model $\Imf'=(\Delta,(\Imc'_i)_{i\in \Zbb})$ of \aformtwo by interpreting the new 
names as outlined in Figure~\ref{fig:ex}. 

\begin{restatable}{theorem}{ThmLTLDLLkromLB}
\label{thm:ltl-dllkrom-lb}
	Satisfiability in \LTLDLLkrom w.r.t.\ interval-rigid names is 
	\ExpSpace-hard.
\end{restatable}
\begin{proof}
	Let \aform be an arbitrary \LTLDLLhorn formula with global CIs, and let $\aform'$ be the formula obtained from \aform by replacing every CI $C\sqcap D\sqsubseteq E$ by a formula $\Psi_{i}$ as above.
	%
	%
	We prove that \aform is satisfiable iff $\aform'$ is satisfiable.
	
	($\Leftarrow$) It is easy to see that this direction holds.
	A model $\Imf'=(\Delta,(\Imc'_i)_{i\in \Zbb})$ of $\aform'$ satisfies, at every time point $j\in\Zbb$, one $\Psi_i$ in such a way that, for each such replacement, we have $e\in A_i^{\Imc'_j}$ if $e\in C^{\Imc'_j}$, $e\in \lprevp \smash{A_i^{\Imc'_j}}$ if $e\in (C\sqcap D)^{\Imc'_j}$ (note the semantics of the interval-rigid names), and thus $e\in (C\sqcap D)^{\Imc'_j}$ implies $e\in E^{\Imc'_j}$.
	
	($\Rightarrow$) We show that every model $\Imf=(\Delta,(\Imc_i)_{i\in \Zbb})$ of \aform can be extended to an interpretation $\Imf'=(\Delta,(\Imc'_i)_{i\in \Zbb})$ such that $\Imf'\models\aform'$. We assume $\Imf'$ to interpret all symbols in \aform the same as \Imf.
	The new names are interpreted such that \dots; $\Imf',-1\models\aformtwo_2$; $\Imf',0\models\aformtwo_0$; $\Imf',1\models\aformtwo_1$; \dots ($*$). 
	This can be achieved by defining $A_i^{
		\smash \Imc'_j}$ for $i\in[0,2]$ and all $j\in\Zbb$ as outlined in Figure~\ref{fig:ex}. Observe that this definition is valid since every element that satisfies one of the new names satisfies it at two consecutive time points. For example, if $e\in A_2^{\Imc'_1}$, which is defined as $(C\sqcap D)^{\Imc_{2}}$, then we also have $e\in A_2^{\Imc'_2}$ since 
	$A_2^{\Imc'_2}=C^{\Imc_{2}}$.
	Lastly, it can readily be checked that ($*$) holds, and thus we get $\Imf'\models\aform'$.
	
We now additionally refer to the proof of \cite[\citethm10]{AKLWZ-TIME07:temporalising} to obtain the \ExpSpace-hardness. More specifically, that proof shows the \ExpSpace-hardness of satisfiability in \LTLDLLhorn by reducing a tiling problem. Without going into further details, it can be seen that the formula describing the tiling is a conjunction of several global CIs and some local ones, but those are in \DLLitecore (see Formulas 22, 26, and 27 in that paper). That is, we can express the tiling in the same way in \LTLDLLkrom w.r.t.\ interval-rigid names.
 Note that the semantics in \cite{AKLWZ-TIME07:temporalising} are over the naturals instead of over the integers, but, as before, this does not change anything and we can use the proof as it is given. 
\end{proof}
Observe that the above CIs in $\Psi_i$ are actually in the Core fragment, and that disjunctions $C\sqsubseteq D\sqcup E$ can be expressed using similar CIs: 
\begin{align*}
C\sqsubseteq{}& A_i&
\lnext A_i\sqsubseteq{}& D&
\lprevp A_{i}\sqsubseteq{}&  E
\end{align*}
Note that these reductions only hold for global CIs; this satisfies our purpose. But, it is likely, that a more complex modeling (e.g., involving counters) could be used to express local CIs in a similar way.

%
In contrast to the above rather negative results, we regain membership in \PSpace by restricting \LTLDLLbool to global CIs. This follows from two facts.
First, recall that interval-rigid concepts can be modeled in the \LTLziDLLbool fragment of \LTLbinDLLbool, and such formulas can be translated in polynomial time into equisatisfiable
LTL$_\DLLite$ formulas (see Sec.~2).
Second, an LTL$_{\DLLitebool}$ formula \aform containing only global CIs can be translated into a propositional LTL formula of size polynomial in $|\aform|$ \cite[\citelem4.3 ff.]{AKRZ-TOCL14:cookbook}.
\cite{AKRZ-TOCL14:cookbook} do not consider assertions combined by arbitrary temporal operators, but since the translation considers assertions directly as propositions and only uses them in a conjunction with other formulas, it is easy to see that it also works for full LTL$_{\DLLitebool}$ restricted to global CIs.

\begin{theorem}\label{thm:ltl-dllbool-global}
	Satisfiability w.r.t.\ interval-rigid names is in \PSpace 
	in \LTLDLLbool restricted to global CIs.\qed
\end{theorem}
\section{Rigid and Interval-Rigid Names in~\DLLite-LTL}

We now study those $\LTLDLLbool$ fragments where LTL operators must not occur within CIs. 
We disregard \LTLbin since it is \ExpSpace-hard, and containment is given in 
Thm.~\ref{thm:alc-results}. 
%
%
First, we describe how interval-rigid concepts can be used to simulate concepts of the form $\lnext C$ and $\lprevp C$ 
in every temporal description logic in which \DLLitecore CIs can be combined by LTL operators;
recall that the operators $\iuntil$ and $\isince$ do not occur in the logics we want to simulate.
We describe the approach for $\lnext$, 
it works similarly for $\lprevp$. Note that the idea is proposed in \cite[\citethm8]{BaBoKOT-FroCoS2017}, 
but that reduction uses the features of \ALC.
The idea is to simulate a concept $\lnext C$ using
six fresh concept names. At every time point, all elements are marked with a flexible concept name $A^i$, $i\in[0,2]$, 
and a $2$-rigid concept $A^{i\ominus_3{}1}_{\lnext C}$ 
transfers the fact that 
$C$ is satisfied to the previous time point, where all elements satisfy $A^{i\ominus_3{}1}$. 
%
\begin{align*}
\boxast\bigvee_{0\le i\le 2}&\Bigg(\left(\top\sqsubseteq A^i\right)\land \bigwedge_{\substack{0\le j\le 2,\\i\neq j}}\left(A^j\sqsubseteq\bot\right)\Bigg)\\
\boxast \bigwedge_{\substack{0\le i\le 2}}\left(\left(\top\sqsubseteq A^i\right)\right.&\rightarrow
\left(A^{i\ominus_3{}1}_{\lnext C} \sqsubseteq C\right)
\land \left(C \sqsubseteq A^{i\ominus_3{}1}_{\lnext C}\right)\land{}
\\\label{eq:next}&\left.
\lnext\left( A^{i\ominus_3{}1}_{\lnext C}\sqsubseteq \bot\right)\land
\lnext\left(\top\sqsubseteq  A^{i\oplus_3{}1}\right) 
\right)\end{align*}	
Regarding some time point,
it is easy to see that, in every model of the formulas, all elements satisfy the same concept $A^i$. Hence, the last CI 
guarantees that all satisfy $A^{i\oplus_3{}1}$ in the moment thereafter.
The CI before ensures that the satisfaction of $C$ at this time point implies that $A^{i\ominus_3{}1}_{\lnext C}$ is satisfied at the previous one (where all elements satisfy $A^{i\ominus_3{}1}$), instead of at the one thereafter.
We thus have that an individual satisfies $\lnext C$ iff it satisfies
$A^i$ and $A^i_{\lnext C}$ for some~$i\in[0,2]$.
Hence, we can simulate a concept $\lnext C$ in a formula \aform by replacing every CI \aaxiom by the conjunction $\bigwedge_{0\le i\le 2}\left(\left(\top\sqsubseteq A^i\right)\rightarrow \aaxiom^i_{\lnext C}\right)$, $\aaxiom^i_{\lnext C}$ is 
obtained from \aaxiom by replacing every outermost concept $\lnext C$ by $A^i_{\lnext C}$, recursively, and forming a conjunction with the above formulas. 

Satisfiability in LTL$_{\DLLitehorn}$ 
can thus directly be reduced to satisfiability in \DLLitehorn-LTL w.r.t.\ interval-rigid concepts; and correspondingly for the Krom fragment. 
\begin{theorem}\label{thm:dllkromhorn-ltl-lb}
	Satisfiability in $\DLLite_{\textit{krom/horn}}$-\upshape{LTL} w.r.t.\ interval-rigid names is \ExpSpace-hard. 
	\qed
\end{theorem}

We therefore also consider the setting without interval-rigid names. Since satisfiability in \LTLDLLkrom is in \PSpace \cite[\citethm5]{AKLWZ-TIME07:temporalising}, we immediately obtain that \DLLitekrom-LTL is also in \PSpace.
Regarding the setting without rigid symbols, we apply the approach that has been proposed for \EL-LTL \cite[\citethm11]{BoT-IJCAI15}. %
It is based on the fact that 
the satisfiability problem can be split into a satisfiability problem in propositional LTL and several (atemporal) satisfiability problems in description logic \cite{BaGL-TOCL12} (see the proof of 
Lem.~4.3 in that paper). 
The former tests the satisfiability of the propositional abstraction $\pa{\aform}$ of \aform, the propositional LTL formula obtained from \aform by replacing the axioms $\aaxiom_1,\dots,\aaxiom_m$ occurring in \aform by propositional variables $p_1,\dots,p_m$, respectively.
The idea is that the worlds $\dots,w_{-1},w_0,w_1,\dots$ in the LTL model of $\pa{\aform}$ characterize the satisfaction of these axioms in the respective description logic interpretations~$\dots,\Imc_{-1},\Imc_0,\Imc_1,\dots$;
a \emph{world} $w_i$ is the set of those propositions that are true in the model at time~$i$.
To obtain~$\Imc_i$ from~$w_i$, we only have to check the satisfiability of the
conjunction of axioms and negated axioms 
induced by~$w_i$ (i.e., $\aaxiom_j$ is negated iff $p_j\not\in w_i$).
Given $k$ as the number of different worlds occurring in the LTL model, it
is hence sufficient to look for $k$ corresponding description logic interpretations.
Regarding the LTL model, note that every satisfiable LTL formula also has a periodic model where the length of the period is exponential in the input and 
whose existence can be checked by using only space of size polynomial in the input  \cite{SiCl85:ltlpspace}, 
in the way of the quasimodel algorithms.
We however cannot regard the set \as of worlds occurring in the model as a whole in a subsequent description logic test since $k=|\as|$ may be exponential; therefore, we integrate it into that algorithm. That is, for every LTL world guessed when iterating over the period, we check the satisfiability of the conjunction of (negated) \DLLitebool axioms induced by it. Note that this problem is in \NP.
This follows from \cite[\citethm8.2]{dllrelations}, where the complexity of satisfiability of conjunctions of CIs and (negated) assertions is stated to be in \NP.
Such a conjunction can be obtained from our conjunctions 
by replacing each negated CI $\lnot(C_1\sqcap\ldots\sqcap C_m\sqsubseteq D_{1}\sqcup\ldots\sqcup D_{n})$ 
by CIs 
$D_{1}\sqcap\overline{C}\sqsubseteq\bot,
\ldots,D_{n}\sqcap\overline{C}\sqsubseteq\bot,$ and assertions 
$C_1(a),\ldots, C_m(a),$ $\overline{C}(a)$; $\overline{C}$ and $a$ are fresh symbols.
\begin{theorem}\label{thm:dllbool-ltl-norigid-ub}
	Satisfiability in \DLLitebool-\upshape{LTL} 
	is contained in \PSpace, without rigid and interval-rigid names.
	\qed
\end{theorem} 

Surprisingly, rigid concepts cause a considerable increase in complexity. Again, the corresponding proof for \EL, a reduction of a tiling problem \cite[\citethm9]{BoT-IJCAI15}, can be adapted. It works similarly w.r.t.\ only interval-rigid concepts since the rigidity of the concepts only needs to last over a number of time points that is exponential in the input. 
%
Containment in \NExpTime can be obtained by integrating the exponentially many description logic satisfiability tests of above into a single one, to guarantee a common interpretation of the rigid symbols. 
Specifically, the $k$ conjunctions of (negated) description logic axioms are considered within one large conjunction, where the flexible symbols are renamed to allow an independent interpretation in each~$\Imc_i$.
%
\begin{theorem}\label{thm:dllhorn-ltl-rigid-lb}
	Satisfiability in $\DLLite_{\textit{horn/bool}}$-\upshape{LTL} w.r.t.\ rigid names is \NExpTime-complete, without interval-rigid names.\qed
\end{theorem}


\section{Conclusions}
We have investigated the complexity of satisfiability in several 
metric temporal extensions of \DLLite and, in particular, 
described the influence of rigid and interval-rigid symbols. 
It turned out that they often yield an expressive power (e.g., conjunction in CIs) and complexity that goes beyond the characteristic properties of the \DLLite fragments and may cause many to yield \ExpSpace complexity results. On the other hand, we identified logics where the complexity of satisfiability is lower than in more expressive DLs, and even obtained \PSpace results for practically interesting formalisms where the axioms that represent domain knowledge hold globally.

In future work, we want to study the influence of role inclusion axioms and of restrictions on the set of allowed temporal operators. 
The problem of query answering and, in that context, data complexity, are also worth to look at, given their relevance in applications.
Interesting would further be a similar study for the DL \EL.
%
%
%


\section{Acknowledgments}
We thank Stefan Borgwardt for many helpful discussions. 
This work is partly supported by the German Research Foundation (DFG) 
within the Cluster of Excellence ``Center for Advancing Electronics Dresden'' (cfaed) in CRC~912. 

%


\bibliographystyle{named}
\bibliography{ref}
%

\end{document}